\documentclass{article}

    \PassOptionsToPackage{numbers, compress}{natbib}

 \usepackage[numbers]{natbib}

\usepackage[utf8]{inputenc} 
\usepackage[T1]{fontenc}    
\usepackage{hyperref}       
\usepackage{url}            
\usepackage{booktabs}       
\usepackage{amsfonts}       
\usepackage{nicefrac}       
\usepackage{microtype}      

\usepackage{amsmath}
\usepackage{amssymb}
\usepackage{mathtools}
\usepackage{amsthm}
\theoremstyle{plain}
\newtheorem{theorem}{Theorem}[section]

\newtheorem{corollary}[theorem]{Corollary}
\theoremstyle{definition}
\newtheorem{definition}[theorem]{Definition}

\theoremstyle{remark}

\usepackage{algorithm}
\usepackage{algpseudocode}
\usepackage{subfiles}
\usepackage{braket}
\usepackage[table,xcdraw]{xcolor}
\usepackage[acronym]{glossaries}
\usepackage{array}
\usepackage{booktabs}
\usepackage{hyperref}
\usepackage{subcaption}

\newtheorem*{remark*}{Remark} 

\usepackage{amsmath}
\usepackage{amssymb}
\usepackage{multirow}
\usepackage{makecell}

\DeclareMathOperator*{\argmin}{arg\,min}

\newcommand{\vt}[1]{
\boldsymbol{#1}
}

\newcommand{\tr}[1]{
  \mathrm{tr}(#1)
}

\newcommand{\cmt}[1]{           
  {\color{red} #1}
}

\title{Zero-shot Quantum Neural Architecture Search}

\author{%
  Tung Dao$^1$, Son N. Tran$^2$, Huynh Thi Thanh Binh$^1$ \\
  $^1$Hanoi University of Science and Technology, Hanoi, Vietnam\\
  $^2$Deakin University, Melbourne, Australia\\
}

\begin{document}

\maketitle

\begin{abstract}
Variational Quantum Algorithms (VQAs) are a leading approach to exploiting near-term quantum hardware, leveraging parameterized quantum circuits and classical optimization to achieve advantage. Despite their promise, the practical deployment of VQAs is challenged by the difficulty of designing quantum circuit architectures that balance expressivity, trainability, and hardware constraints. Existing evolutionary-based quantum neural architecture search methods address these challenges but suffer from high computational costs due to repeated training of candidate circuits. In this work, we identify a setting in which the Gram matrix of the Quantum Neural Tangent Kernel converges. Building on this observation, we design a zero-shot surrogate model to estimate candidate performance without full training, significantly accelerating the architecture search process. Using this surrogate, we propose MZeQAS, a Monte Carlo Tree Search (MCTS)-based Zero-Shot Quantum Neural Architecture Search framework for VQAs. By integrating proxy-based performance estimation with MCTS exploration, MZeQAS efficiently discovers high-performing architectures. Experimental results demonstrate that MZeQAS outperforms existing approaches in terms of both search efficiency and solution quality, providing a scalable and effective framework for advancing VQA deployment on noisy intermediate-scale quantum devices.
\end{abstract}

\section{Introduction}
The emergence of quantum computing has introduced a new computational paradigm that exploits quantum-mechanical effects, motivating new approaches to computation and learning. Motivated by recent progress in machine learning, quantum neural networks (QNNs) have been proposed as a promising framework for applying quantum computing to learning problems. Among existing approaches, Variational Quantum Algorithms (VQA) have emerged as a practical paradigm for near-term quantum devices, commonly referred to as Noisy Intermediate-Scale Quantum (NISQ) hardware \cite{cerezo2021variational, preskill2018quantum}. The performance of VQAs is highly sensitive to the underlying circuit architecture. Qubit mapping, circuit depth, gate composition, and parameterization jointly influence expressivity, trainability, and noise sensitivity. This motivates the development of automated approaches—commonly referred to as quantum architecture search (QAS)—that can efficiently design high-performing VQA architectures under realistic hardware constraints.

However, identifying optimal circuit structures for a given computational task remains a non-trivial challenge, particularly when balancing expressivity. Traditionally, researchers have relied on problem-specific insights to design problem-inspired ansätze \cite{kandala2017hardware, lee2018generalized}. More recently, automatic quantum circuit architecture search has emerged as a promising alternative, leveraging classical search algorithms \cite{duong2022quantum,Zhang_npredictor_2021,Zhang_differential_2022,kuo2021_qas_rl}. Among these methods, evolutionary algorithms (EAs) \cite{wang2022quantumnas, zhu2022adaptive,li2023eqnas} and Monte Carlo Tree Search (MCTS) \cite{lipardi2025quantum} have become prominent paradigms, owing to their effectiveness in navigating large, discrete, and non-differentiable circuit search spaces. Despite their promising results, these approaches often suffer from significant computational bottlenecks, as they require repeated training to evaluate the performance of a large number of candidate quantum circuits. To mitigate this issue, recent studies have proposed surrogate-based evaluations, including one-shot training with weight sharing \cite{wang2022quantumnas} and learned neural performance predictors \cite{zhang2021neural}. While these techniques reduce search cost, they still depend on partial training or the availability of additional learned models, which introduces new sources of bias and computational complexity. Meanwhile, only limited effort has been devoted to training-free architecture search for VQAs. Early attempts, such as (Training-Free Quantum Architecture Search) TF-QAS \cite{he2024training}, demonstrate that training-free proxies can efficiently identify circuits with desirable structural properties. However, this method primarily emphasizes circuit expressiveness and structural complexity, rather than task-specific learning behavior.

In this paper, we propose Asymptotic Minimum Eigenvalue Surrogate (AMES), a zero-shot proxy for estimating the loss of quantum neural networks. Our proxy evaluates the quality of a variational ansatz by estimating the convergence behavior of the training loss through properties of the quantum neural tangent kernel. We show that the loss convergence can be approximated independently of the training iteration, yielding a closed-form convergence metric that can be computed without performing any training. Building on this proxy, we develop an efficient and fully training-free framework for QNN architecture search that directly leverages AMES. This design enables rapid performance estimation without full training, significantly reducing search cost. Our approach integrates this proxy-based evaluation with MCTS, which iteratively refines circuit candidates by expanding the search space while focusing on promising architectures. To further improve scalability, we incorporate {random batch sampling, which reduces the computational complexity of tree expansion while preserving exploration quality and search effectiveness}. We refer to our method as MZeQAS. Experimental results show that MZeQAS consistently identifies high-performing quantum circuit architectures more efficiently than existing methods, achieving faster search times while delivering superior solution quality. This demonstrates that MZeQAS effectively balances search efficiency with the discovery of high-quality circuits, providing a scalable solution for quantum architecture design. The contributions of this paper are summarized as follows:
\begin{itemize}
    \item We show that under certain weight initialization distributions, the behavior of the Gram matrix in the Quantum Neural Tangent Kernel asymptotically approaches a well-defined expected value.
    \item We introduce AMES, a surrogate loss function for efficient quantum neural architecture evaluation, providing evidence of their effectiveness in reducing computational cost without compromising search quality.  
    \item We propose MZeQAS, a novel zero-shot quantum neural architecture search framework that integrates a zero-shot proxy model with MCTS and incorporates random batch sampling for enhanced scalability.  
    \item We demonstrate, through extensive experiments, that MZeQAS outperforms existing state-of-the-art methods in both search efficiency and discovered circuit quality for variational quantum algorithms.  
\end{itemize}

The remainder of this paper is organised as follows. The next section reviews the related literature, followed by a detailed description of the proposed approach. The experimental section then outlines the setup and presents an analysis of the results. Finally, the paper concludes with a summary of findings and potential directions for future research.

\section{Related Work}
\label{sec:rw}
\subsection{Quantum Architecture Search}
Recent quantum neural architecture search methods increasingly borrow ideas from classical neural architecture search to automate variational circuit design. Bayesian optimization–based approaches guide the search using quantum circuit distance that capture the behavior of the gates over the state space \cite{duong2022quantum}, while neural performance predictors estimate circuit quality without full training \cite{Zhang_npredictor_2021}. Differentiable search methods relax discrete architectural choices to enable gradient-based optimization \cite{Zhang_differential_2022}, and reinforcement learning frameworks model circuit construction as a sequential decision process \cite{kuo2021_qas_rl}.

Among NAS-inspired approaches, evolutionary-based methods remain among the most commonly used strategies for effectively discovering effective quantum circuit architectures. QuantumNAS \cite{wang2022quantumnas}, for instance, applies an evolutionary algorithm combined with a front-sampling technique that samples only the initial blocks and gates during the search. This approach reduces variance and improves convergence toward higher-quality solutions. Similarly, EQNAS \cite{li2023eqnas} adopts a quantum evolutionary algorithm, which leverages quantum principles for both encoding individuals and performing genetic operations within Evolutionary Algorithm (EA). These studies demonstrate the effectiveness of EQ-based methods for QAS. However, like many population-based approaches, they are prone to becoming trapped in local optima, as low-performing individuals are frequently discarded in favor of higher-performing ones. This neglects the potential of these individuals. More recently, researchers have explored MCTS for QAS \cite{wang2023automated}, where each node in the search tree represents a quantum circuit and the tree is incrementally expanded by adding quantum gates. Building on this work, \cite{lipardi2025quantum} enhanced MCTS by introducing progressive widening - Progressive Widening
enhanced Monte Carlo Tree Search (PWMCTS), further improving its search efficiency. However, since PWMCTS modifies only one gate at a time, it struggles to generate more complex quantum circuits in challenging tasks due to its fixed tree depth.

Retraining each configuration in QAS to obtain a fitness score is computationally expensive and time-consuming, especially given the limited access to quantum hardware. To mitigate this, several studies have proposed using surrogate models to estimate quantum circuit performance \cite{zhang2021neural, he2024meta, he2024training, he2024quantum, wang2022quantumnas}. For instance, QuantumNAS \cite{wang2022quantumnas} employs a weight-sharing technique, where a parameterized SuperCircuit is trained once and candidate SubCircuits inherit the corresponding weights, avoiding the need for full retraining. To eliminate any pretraining prior to evaluating the final circuit, TF-QAS \cite{he2024training} introduces a two-stage zero-shot proxy that estimates circuit performance based on the number of paths and the expressibility of the quantum circuit. Experimental results from TF-QAS demonstrate the potential of zero-shot proxies in QAS by significantly reducing search time while still identifying high-performing circuits. However, TF-QAS only considers the expressiveness of the circuit and does not account for how the relationship between the network and its loss function evolves during training. Moreover, it has been shown in the literature that expressiveness is directly linked to barren plateau in VQA \cite{larocca2025barren}.

\subsection{Zero-shot Neural Architecture Search}
Zero-shot proxies for evaluating neural network architectures are widely adopted in classical Neural Architecture Search (NAS) \cite{Lin_2021_ICCV, jiang2023meco, lizico, dong2025parzc}. They eliminate the need for explicit training, thereby significantly improving search efficiency. Zero-shot NAS methods typically rely on gradient information or loss landscape properties to construct surrogate objectives. For example, \cite{chen2021neural, shu2022unifying, zhu2022generalization} derive zero-shot proxies grounded in NTK theory \cite{DuZPS19}. In a subsequent work, \cite{jiang2023meco} propose a surrogate based on the minimum eigenvalue of the Pearson correlation matrix computed from feature maps. Overall, while a substantial body of work in NAS exploits loss convergence behavior or loss landscape analysis to design zero-shot proxies, similar approaches remain largely unexplored in QAS. This discrepancy highlights a significant research gap that our work aims to address.

\section{Preliminaries}
\subsection{Quantum Neural Networks}
\label{subsec: qnn}

A QNN begins by encoding a classical input vector $\mathbf{x}$ into an $d$-qubit quantum state within a Hilbert space, typically using a feature map of the form: 
$\ket{\psi(\mathbf{x})} = F(\mathbf{x})\ket{0}^{\otimes d}$, 
where $F(\mathbf{x})$ is a unitary operator that transforms classical data into a quantum representation of the input. This input state, represented as a density matrix $\boldsymbol{\rho}$, is then processed through a parameterized quantum circuit, commonly referred to as an ansatz, which consists of a sequence of rotation and entangling gates. The ansatz is defined by a set of trainable parameters $\boldsymbol{\theta}$, forming a unitary transformation expressed as $
U(\boldsymbol{\theta}) = \prod_{i=0}^{L} U_i(\theta_i)$, where each $U_i(\theta_i) := U_i\ \exp{(-i\theta\vt{H}_i)}$ represents a gate that depends on the corresponding parameter $\theta_i$ and Hermitian operator $\vt{H}_i$ while $L$ and $p$ denote the total number of layers and parameters, respectively. This unitary evolves the feature-encoded quantum state into a new state that captures the model’s learned representation of the input. The model output is obtained by measuring the expectation value of a Hermitian observable $M$, given by
\begin{equation}
f(\boldsymbol{\theta},\vt{x}) = \tr{\boldsymbol{\rho}U^\dagger(\boldsymbol{\theta})MU(\boldsymbol{\theta})}.
\label{eq: qnn}
\end{equation} 
For notational convenience, we define $\mathbf{M}(\boldsymbol{\theta}):=U^\dagger(\boldsymbol{\theta})MU(\boldsymbol{\theta})$. Throughout this paper, we restrict our attention to Pauli measurement operators $M$. Pauli measurements are tensor product of Pauli matrices which are $2 \times  2$ unitary Hermitian matrices $\sigma_X = 
    \begin{pmatrix}
    0 & 1 \\
    1 & 0 \\
    \end{pmatrix}$, 
$\sigma_Y = \begin{pmatrix}
    0 & -i \\
    i & 0 \\
    \end{pmatrix}$,
$\sigma_Z = \begin{pmatrix}
    1 & 0 \\
    0 & -1 \\
    \end{pmatrix}.$
  
Let us consider a dataset $\mathcal{D} = {(\mathbf{x}^{(m)}, y^{(m)})}_{m=1}^{|\mathcal{D}|}$, where each $\mathbf{x}^{(m)} \in \mathbb{R}^d$ represents a input vector and $y^{(m)} \in \mathcal{Y}$ is the corresponding label, drawn from a finite label set $\mathcal{Y}$. Similar to classical neural networks, the goal is to train a QNN to learn a function that maps input vectors to their associated outputs by minimizing a suitable loss function over the dataset, e.g the difference between the expected outputs and the labels  
\begin{equation}
    \mathcal{L}(\boldsymbol{\theta})=\frac{1}{2|\mathcal{D}|}\sum_{(\vt{x}^{(m)},y^{(m)})\in \mathcal{D}} (f(\boldsymbol{\theta},\vt{x}^{(m)})-y^{(m)})^2 
\end{equation}

The aim of the training process is to estimate the parameters, such that $\boldsymbol{\theta}^* = \argmin_{\boldsymbol{\theta}} \mathcal{L}(\boldsymbol{\theta})$. Current training methods for QNN rely on classical optimization techniques, utilizing gradient descent to search for the optimal parameters. This requires iterative parameter updates to minimize the loss, based on the gradient of the loss function over the parameters. We denote $\mathcal{L}(\boldsymbol{\theta}(t))$ and $\theta(t)$ as the loss function and parameters after $t$ training iterations.

\section{MCTS-based Zero-shot Quantum Neural Architecture Search}
\begin{figure*}[!ht]
    \centering
    \includegraphics[width=\linewidth]{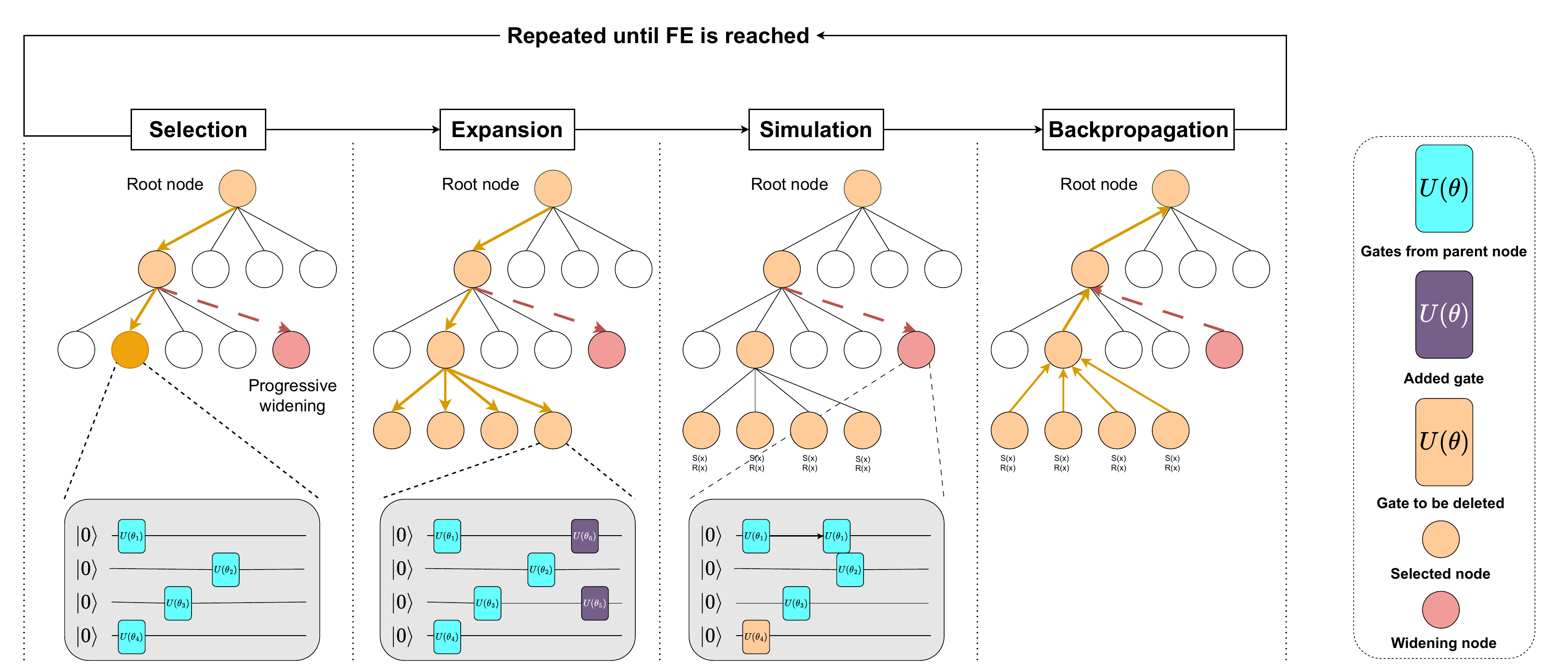}
    \caption{Each iteration of MZeQAS consists of four stages, from left to right: Selection, Expansion, Simulation, and Backpropagation. After selecting the most promising node, MZeQAS adds new gates (colored purple) to create a new node. Progressive widening may randomly delete gate(s) from the parent circuit (colored orange) to further diversify the search space. This process is repeated until the function evaluation (FE) budget is reached.}
    \label{fig:mcts}
\end{figure*}

\subsection{Quantum Neural Architecture Search}

Let $\mathcal{V}$ denote the architecture search space, and for each $v \in \mathcal{V}$ let $\Theta(v)$ denote the parameter space induced by architecture $v$. The objective of quantum architecture search is to identify an optimal architecture $v^*$ such that
\begin{equation}
v^* = \arg\min_{v \in \mathcal{V}} g(v).
\end{equation}
Candidate quantum circuits are generated by a search algorithm, and each candidate architecture $v$ is evaluated using a fitness function $g(v)$ that measures its quality. In prior QAS studies, the fitness function $g(v)$ is commonly defined using the training loss of a fully optimized model, for example $\mathcal{L}(\theta(t))$ where $\theta(t) \in \Theta(v)$, or by employing a surrogate objective that approximates this quantity to reduce computational cost.

\label{sec:zero-shot}
\subsection{Loss Function Convergence of QNN}
In this section, we revisit the neural tangent kernel (NTK) analysis of the training dynamics of QNNs with Pauli measurements \cite{you2023analyzing}. It is shown that when the number of trainable parameters \(p\) is sufficiently large and the minimum eigenvalue of the QNN Gram matrix \(\mathbf{K}_{\mathrm{asym}}(t)\) is lower bounded by a positive constant \(\lambda_0\), the training loss of a periodic ansatz (defined in Appendix~\ref{appendix: periodic ansatz}) converges linearly with high probability under random initialization:
\begin{equation}
\label{eq:linear-conv}
\mathcal{L}(\vt{\theta}(t)) \leq \mathcal{L}(\vt{\theta}(0)) \exp\!\left(-\frac{\lambda_0 t}{2}\right).
\end{equation}
The entries of the asymptotic kernel \(\mathbf{K}_{\mathrm{asym}}(t)\) are given by
\begin{equation}
\label{eq:kernel}
\bigl(K_{\mathrm{asym}}(t)\bigr)_{kj}
:= 
\operatorname{tr}(i\!\left[\mathbf{M}(\vt{\theta}(t)), \vt{\rho}_k\right]
\, 
i\!\left[\mathbf{M}(\vt{\theta}(t)), \vt{\rho}_j\right]
).
\end{equation}

The time-dependent matrix $\mathbf{K}_{\mathrm{asym}}(t)$ governs the evolution of the residual vector, which captures the discrepancy between model predictions and target labels during training. As a result, analyzing the training behavior of a QNN can be reduced to studying the spectral properties of $\mathbf{K}_{\mathrm{asym}}(t)$, with the minimum eigenvalue directly controlling the convergence rate. This observation would provide a principled link between circuit architecture and optimization performance. However, as evident from Eq.~\eqref{eq:kernel}, the kernel remains training-dependent and, in general, requires tracking the parameter trajectory $\vt{\theta}(t)$, which limits its direct applicability for efficient architecture evaluation.

\subsection{Asymptotic Minimum Eigenvalue Surrogate}

In this section, we show how to construct a training-free surrogate function based on the kernel defined in Eq.~\eqref{eq:kernel}. To eliminate the need for explicit parameter optimization, we leverage the overparameterization setting of QNNs, under which the training dynamics admit a kernel-based characterization. In particular, we show that when the number of trainable parameters is sufficiently large, the time-dependent kernel $\mathbf{K}_{\mathrm{asym}}(t)$ approximately concentrates around a time-independent limit that depends only on the circuit architecture and the parameter distribution.
\begin{theorem}
\label{theorem: t=inf}
Consider an overparameterized QNN with $p$ trainable parameters. As $p \to \infty$, the time-dependent Gram matrix $\mathbf{K}_{\mathrm{asym}}(t)$ converges entry-wise to a time-independent limiting matrix $\mathbf{K}_{\mathrm{asym}}^{\infty}$, i.e.,
\[
\mathbf{K}_{\mathrm{asym}}(t) \;\approx\; \mathbf{K}_{\mathrm{asym}}^{\infty}.
\]
Each entry of the limiting Gram matrix is given by
\begin{equation}
\label{eq:kasym_approx}
(\mathbf{K}_{\mathrm{asym}}^{\infty})_{kl}
= \mathbb{E}_{\boldsymbol{\theta}\sim \mathcal{N}(\mathbf{0},\vt{I}/p)}
\!\left[
\operatorname{tr}\!\left(
i[M(\boldsymbol{\theta}), \rho_k]\,
i[M(\boldsymbol{\theta}), \rho_l]
\right)
\right],
\end{equation}
where $\mathcal{N}(\mathbf{0},\vt{I}/p)$ denotes the standard multivariate Gaussian distribution over circuit parameters.
A proof is provided in Appendix~\ref{appendix:proof}.
\end{theorem}

We leverage the convergence result in Equation 4 together with Theorem~\ref{theorem: t=inf} to construct the Asymptotic Minimum Eigenvalue Surrogate (AMES). Let $\lambda_{\min}$ denote the minimum eigenvalue of the asymptotic Gram matrix $\mathbf{K}_{\mathrm{asym}}^{\infty}$. Using Corollary \ref{cor:lambda-equal}, the minimum eigenvalue $\lambda_0$ appearing in Eq.~\eqref{eq:linear-conv} can be approximated by $\lambda_{\min}$.  Combined with the linear convergence property, this approximation provides the following upper bound on the converged training loss:
\begin{equation}
\label{eq:ames_upper_bound}
\mathcal{L}(\boldsymbol{\theta}^*)
\;\leq\;
\mathcal{L}(\boldsymbol{\theta}(0)) \,
\exp\!\Big(-\frac{\lambda_{\min}t}{2}\Big),
\end{equation}
where $\boldsymbol{\theta}^*$ denotes the optimized parameters after convergence.

This bound characterizes the worst-case convergence behavior of a variational quantum circuit: architectures associated with a lower upper bound will be more likely to attain a smaller training loss.

In the overparameterized setting of QNNs, the Gram matrix $\mathbf{K}_{\mathrm{asym}}^\infty$ becomes highly concentrated. That is, as the number of trainable parameters $p$ grows, the variance of the entries of $\mathbf{K}_{\mathrm{asym}}^\infty$ under random initialization decreases. This is analogous to classical Neural Tangent Kernel results, where the NTK converges to its expectation in the infinite-width limit \cite{jacot2018neural}. Therefore, sampling a single $\boldsymbol{\theta} \sim \mathcal{N}(\mathbf{0}, \vt{I}/p)$ is sufficient to approximate $\mathbf{K}_{\mathrm{asym}}^\infty$ with high probability.

Finally, since overparameterization is known to improve kernel conditioning and training dynamics, we explicitly encourage architectures with larger parameter counts by scaling the surrogate with the inverse of the number of circuit parameters, denoted by $p_c$. The resulting AMES score is defined as

\begin{equation}
    g(c) = \frac{1}{p_c} \mathcal{L}(c(0)) 
    \exp\left(-\frac{\lambda_{\min}}{2}\right).
    \label{eq:surrogate}
\end{equation}

which balances trainability and model capacity while remaining entirely training-free.

Note that although the \autoref{theorem: t=inf} and Eq.~\eqref{eq:linear-conv} were originally constructed for Mean Squared Error (MSE) loss,  \cite{Goodfellow-et-al-2016} points out that if the model follows Gaussian distribution then optimizing for MSE is equivalent to optimizing for Cross Entropy (CE) loss. Therefore, our surrogate model can also be applied to classification tasks.

\subsection{Monte Carlo Tree Search}

MCTS offers a robust solution to limitations of evolutionary methods, such as their tendency to get trapped in local optima, enabling the generation of more complex quantum circuits. MCTS mitigates this by preserving all generated circuit structures within its tree, allowing the search to backtrack and explore alternative paths when encountering stagnation. In our approach, each node in the MCTS tree represents a quantum circuit, with the root node being an empty circuit. Quantum gates are progressively added as the tree's depth increases. Furthermore, to overcome the limitations in PWMCTS, we implement batch sampling during each expansion step, which facilitates the generation of more intricate circuits. Figure \ref{fig:mcts} illustrated our method MZeQAS. 
\\
\textit{Selection}: At the selection stage, MZeQAS calculate the $UCT$ for each child node $\hat{c} \in children(c)$ of node $c$ using the following equation:
\begin{equation}
    UCT(\hat{c}) = \left(\frac{Q(\hat{c}) - q_{min}}{q_{max} - q_{min}} \right) + \beta\sqrt{\frac{ln(N(c) + 1}{N(\hat{c})}}).
    \label{eq:uct}
\end{equation}
In this formula, $q_{max}$ and $q_{min}$ are the highest and lowest $Q$ values encountered anywhere within the MCTS tree. Furthermore, MZeQAS employed the \textit{Exploration decay} techniques \cite{zhengmonte}. This technique dynamically adjusts $\beta$ to prioritize exploration in the early phases of the search and gradually shift emphasis towards exploitation. The decay of $\beta$ is determined by:
\begin{equation}
    \beta = \beta_0*\frac{FE-fe}{FE},
    \label{eq: decay}
\end{equation}
where $FE$ denotes the maximum number of function evaluations and $fe$ is the current number of evaluations performed.
\vspace{5pt}
\\
\textit{Expansion}: After identifying the most promising node $c^*$, MZeQAS proceeds with expansion. The available gates constitute the search space $S = \{s_1, s_2,..., s_G\}$, with $R = \{r_1, r_2,...,r_G\}$ be the probability of choosing each gate $s_i$. Gates are randomly sampled according to these probabilities. Subsequently, a variable number ($H$, $1 \le H \le L$) of these selected gates are appended to $c^*$. The above steps are repeated $B$ times to generate the new set of children $children(c^*)$.
\vspace{5pt}
\\
\textit{Simulation}: This step evaluates the performance of each newly created node. In MZeQAS, the performance of each node, $g(.)$, is derived from Eq.~\ref{eq:surrogate}. We set the quality value to the negative of AMES, $Q(c) \gets -g(c)$, record their visit count $N(c) \gets 1$ and update $q_{max}$ and $q_{min}$.
\vspace{5pt}
\\
\textit{Backpropagation}: Finally, we traverse upward from the current leaf to the tree root. The quality value and visit count will be updated according to these equations:
\begin{equation}
    Q(c) \gets \max_{\hat{c} \in Children(c)} Q(\hat{c}),
\end{equation}
\begin{equation}
    N(c) \gets \sum_{\hat{c} \in Children(c)} N(\hat{c}).
\end{equation}
\vspace{5pt}
\\
\textit{Progressive widening}: This technique, introduced in \cite{coulom2007computing} and subsequently adopted in QAS task by \cite{lipardi2025quantum}, is employed to strategically re-explore non-leaf nodes. A new child node is added to a node $c$ when its visitation count satisfied the condition,
\begin{equation}
    \lfloor N(c)^\alpha \rfloor \ge |children(c)|.
\end{equation}
effectively increasing the branching factor for frequently visited paths. Additionally, MZeQAS incorporates a 'delete' operator. This operator creates a new child node $\hat{c}$ by randomly removing gates from its parent. 

\paragraph{Connection between MZeQAS's depth in batch-sampling and AMES} Let $D$ denote the maximum depth of an MCTS tree. In the batch-sampling strategy, the expected number of variables added at each expansion step is the expected value of $H$, $E[H] = E\left[\frac{L}{2}\right],$,
which implies that the quantum circuit grows at a faster rate ($O(\frac{L}{2})$) and can reach the overparameterized regime more quickly. In contrast, in the standard MCTS approach, only one gate is added per expansion step, resulting in a maximum circuit size of $O(D)$. Therefore, the batch-sampling strategy employed in MZeQAS is particularly well-suited for the surrogate function AMES and beneficial to QAS.

\section{Experiment}

\subsection{Tasks \& Datasets}

This section evaluates the performance of MZeQAS in comparison with hand-crafted QNN and other QAS-based methods. We evaluate our framework under three scenarios: classical image classification, quantum phase recognition and hybrid-classic architecture.

\noindent\textbf{Classical Image Classification}. We evaluate the performance of MZeQAS on classical task with the image classification dataset MNIST \cite{deng2012mnist}, Fashion-MNIST \cite{xiao2017fashion} and Vowel dataset \cite{vowel}. MZeQAS is compared against QuantumNAS \cite{wang2022quantumnas}, PWMCTS \cite{lipardi2025quantum} and EQNAS \cite{li2023eqnas}. We implemented MZeQAS using the TorchQuantum library \cite{wang2022quantumnas}. The comparison with EQNAS uses two-qubit gates: \textit{XX}, \textit{YY}, and \textit{ZZ}. A detailed breakdown of dataset and their corresponding label, input size and non-linear data encoder is provided in \autoref{sec: expsetup}. More information about experiment setup can be found in \autoref{sec: expsetup} in the Supplementary Material. 

\noindent\textbf{Quantum Phase Recognition (QPR)}. Quantum Phase Recognition is a classification problem that aims to identify phases of quantum matter, a fundamental task in condensed matter physics~\cite{sachdev1999quantum, broecker2017machine, carrasquilla2017machine, ebadi2021quantum, sachdev2023quantum}. In this work, we adopt the same QPR task, data generator, and QCNN implementation from~\cite{hurunderstanding}, which are based on the generalized cluster Hamiltonian. Our main objective is to evaluate whether MZeQAS can discover circuits that approximate the performance of human-designed architectures. To this end, we generate 1,000 instances for training and 100 instances for testing. Further details on the experimental configuration are provided in \autoref{sec: expsetup}.

\noindent\textbf{Hybrid Quantum-Classical Architecture}. 
To examine whether MZeQAS remains effective beyond standalone QNN classifiers, we incorporate it into the Quantum-Train (QT) framework~\cite{liu2025quantum}. In QT, the QNN serves as a compact parameter generator: its measurement outcomes are passed through a lightweight classical network to produce the weights of a classical Convolutional Neural Network (CNN), which is then used for final prediction. This setting allows us to evaluate whether MZeQAS can discover useful QNN modules within a hybrid quantum-classical learning pipeline. For this task, MZeQAS is used to search for the QNN component. Additional experimental details are provided in \autoref{sec: expsetup}.

\subsection{Classical image classification}

\begin{table*}[ht]
\scriptsize
\centering
\caption{Mean accuracy and standard deviation on different image classification datasets and number of classes}
\label{tab:accuracy_comparison}
\begin{tabular}{l|cccc|cc}
\hline
\multirow{3}{*}{\textbf{Datasets}}
& \multicolumn{4}{c|}{\textbf{U3 + CU3}} 
& \multicolumn{2}{c}{\textbf{XX + YY + ZZ}} \\
\cmidrule(lr){2-5} \cmidrule(lr){6-7}
& \textbf{\makecell{MZeQAS \\ (ours)}}
& \textbf{\makecell{Quantum-\\NAS}}
& \textbf{PWMCTS}
& \textbf{\makecell{TF-\\QAS}}
& \textbf{EQNAS} & \textbf{\makecell{MZeQAS \\ (ours)}} \\
\hline
MNIST-2   & 85.6 \(\pm\) 0.2 & 84.0 \(\pm\) 0.3 & 78.7 $\pm$ 1.0 & \cellcolor{gray!30}\textbf{86.2 \(\pm\) 0.2} & 85.1 $\pm$ 1.9 & 86.0 $\pm$ 0.5\\
MNIST-4   & \cellcolor{gray!30}\textbf{77.3 $\pm$ 0.8} & 70.5 $\pm$ 0.2 & 65.9 $\pm$ 1.0 & 56.7 $\pm$ 0.3 & - & 76.0 $\pm$ 1.3\\
MNIST-10  & \cellcolor{gray!30}\textbf{53.4 $\pm$ 2.2} & 51.0 $\pm$ 3.7 & 14.4 $\pm$ 0.8 & 26.7 $\pm$ 2.7 & - & 46.0 $\pm$ 3.6\\
\hline
Fashion-2 & 83.0 $\pm$ 0.8 & \cellcolor{gray!30}\textbf{86.7 $\pm$ 1.8} & 85.5 $\pm$ 1.2 & 81.3 $\pm$ 1.3 & 55.3 $\pm$ 1.5 & 79.2 $\pm$ 1.7\\
Fashion-4 & \cellcolor{gray!30}\textbf{78.5 $\pm$ 0.7} & 74.1 $\pm$ 5.8 & 57.6 $\pm$ 3.5 & 68.7 $\pm$ 2.2 & - & 77.0 $\pm$ 1.2\\
Fashion-10& \cellcolor{gray!30}\textbf{53.7 $\pm$ 0.9} & 48.0 $\pm$ 2.6 & 8.5 $\pm$ 0.3 & 41.6 $\pm$ 2.6 & - & 45.0 $\pm$ 4.5 \\
\hline
Vowel-4   & \cellcolor{gray!30}\textbf{51.0 $\pm$ 2.3} & 39.8 $\pm$ 3.6 & 42.3 $\pm$ 4.2 & 38.1 $\pm$ 3.7 & - & 39.8 $\pm$ 0.9\\
\hline
\end{tabular}

\end{table*}

\begin{table}[h]
\caption{Comparison of runtime, including search and training, between MZeQAS and QuantumNAS. Runtime is measured in hours.}
\label{tab:runtime}
\centering
\begin{tabular}{llcccc}
\hline
\textbf{Dataset} & \textbf{Method} & \textbf{Search} & \textbf{Training} & \textbf{Total} & \textbf{\%Imp.} \\ 
\hline

\multirow{2}{*}{MNIST-2}
& QuantumNAS      & 0.114 & 0.131 & 0.245 & -- \\
& MZeQAS (ours)   & \cellcolor{gray!30}\textbf{0.040} 
                  & \cellcolor{gray!30}\textbf{0.107} 
                  & \cellcolor{gray!30}\textbf{0.147} 
                  & \cellcolor{gray!30}\textbf{40.0\%} \\
\hline

\multirow{2}{*}{MNIST-4}
& QuantumNAS      & 0.168 
                  & \cellcolor{gray!30}\textbf{0.115} 
                  & 0.283 
                  & -- \\
& MZeQAS (ours)   & \cellcolor{gray!30}\textbf{0.058} 
                  & 0.126 
                  & \cellcolor{gray!30}\textbf{0.184} 
                  & \cellcolor{gray!30}\textbf{34.9\%} \\
\hline

\multirow{2}{*}{MNIST-10}
& QuantumNAS      & 0.717 
                  & \cellcolor{gray!30}\textbf{0.853} 
                  & 1.570 
                  & -- \\
& MZeQAS (ours)   & \cellcolor{gray!30}\textbf{0.439} 
                  & 0.955 
                  & \cellcolor{gray!30}\textbf{1.394} 
                  & \cellcolor{gray!30}\textbf{11.2\%} \\
\hline

\end{tabular}
\end{table}

\paragraph{Accuracy comparison} \autoref{tab:accuracy_comparison} presents a detailed comparison of MZeQAS, QuantumNAS, PWMCTS, TF-QAS and EQNAS on three benchmark datasets: MNIST, Fashion-MNIST, and Vowel, each with varying numbers of classes. The experiment is conducted in the 'U3' and 'CU3' search space, where 'CU3' gates are arranged in ring connections such that for every qubit $i$, a CU3($i$, $i+1$) operation can be applied. Meanwhile, we also conduct comparison with EQNAS on a different set of gate, 'XX', 'YY' and 'ZZ', that is used in the original paper \cite{li2023eqnas}. Since EQNAS uses a single read-out qubit as output, we only run it on two-class classification datasets. For fair comparison, we set the same timeout for all approaches to 6 hours. MZeQAS consistently outperforms QuantumNAS and PWMCTS across most datasets, demonstrating both strong accuracy and robustness. The only exception is Fashion-2, where QuantumNAS and PWMCTS achieve slightly higher accuracies (86.7\% and 85.5\%, respectively) compared to MZeQAS at 83.0\%. However, as the number of classes and the number of qubits increases, PWMCTS exhibits significant scalability challenges. While it performs comparably to the other methods on smaller datasets, its accuracy drops sharply on MNIST-10 and Fashion-10, highlighting its difficulty in constructing deeper or more complex circuits within a fixed tree depth. By contrast, MZeQAS maintains consistently strong performance across all datasets and class settings. This stability can be attributed to its batch-sampling strategy, which mitigates the search-space exploration limitations observed in previous MCTS-based methods such as PWMCTS. Moreover, in comparison with another zero-shot proxy QAS approach, TF-QAS, our method performs better in most task except for MNIST-2. These findings indicate that MZeQAS not only scales better to more complex search spaces but also retains high accuracy as the search task becomes increasingly challenging. Furthermore, compared with EQNAS in a different search space, MZeQAS also achieves higher accuracy on binary classification task although it's lower than the result using U3 and CU3 search space.

\paragraph{Runtime analysis} Table~\ref{tab:runtime} shows that MZeQAS consistently achieves lower runtime compared to QuantumNAS across all datasets. This runtime analysis was run with a device contains 96-core Intel(R) Xeon(R) CPU @ 2.00GHz and 335GB of RAM. A detailed breakdown of the runtime components shows that the majority of the improvement comes from faster search time. For instance, in MNIST-2, the search phase drops from 0.114 to 0.040, while the training time remains decently close (0.131 vs. 0.107). This pattern is consistent across MNIST-4 and MNIST-10, highlighting the efficiency of MZeQAS during the search phase. Unlike QuantumNAS, MZeQAS does not require the evaluation of a SuperCircuit, thereby eliminating the need to train a complex overarching circuit before searching for high-performing circuits. By streamlining the search phase, MZeQAS is capable of identifying high-performing circuits more quickly, making it advantageous in many scenarios. 

\subsection{Quantum Phase Recognition}

\begin{table}[ht]
    \centering
    \caption{Accuracy comparison on the clustering Quantum Phase Recognition task. Mean accuracy and standard deviation are reported, with the number of parameters shown in parentheses.}
    \label{tab:acc_qpr}
    \begin{tabular}{c|c|c|c|c}
        \hline
        \textbf{Task}
        & \textbf{SEL}
        & \textbf{\makecell{QCNN\\(1 Layer)}}
        & \textbf{\makecell{QCNN\\(9 Layers)}}
        & \textbf{MZeQAS} \\
        \hline

        \makecell{Clustering\\(8 qubits)}
        & 48.0 $\pm$ 0.025
        & \makecell{71.5 $\pm$ 0.016\\(195)}
        & \makecell{88.2 $\pm$ 0.015\\(1635)}
        & \makecell{\textbf{88.2 $\pm$ 0.019}\\(1069)} \\
        \hline

        \makecell{Clustering\\(10 qubits)}
        & 51.0 $\pm$ 0.046
        & \makecell{78.5 $\pm$ 0.015\\(225)}
        & \makecell{\textbf{88.7} $\pm$ 0.008\\(1920)}
        & \makecell{88.0 $\pm$ 0.019\\(1001)} \\
        \hline
    \end{tabular}
\end{table}

The one-dimensional generalized cluster Hamiltonian is defined as $H(J_1, J_2) = \sum_{j=1}^{N} 
\left( Z_j - J_1 X_j X_{j+1} - J_2 X_{j-1} Z_j X_{j+1} \right)$ where $X_j$ and $Z_j$ are Pauli operators acting on site $j$ while $J_1$ and $J_2$ are tunable coupling parameters. These parameters determine the phase of the system, which may be ferromagnetic, antiferromagnetic, symmetry-protected topological (SPT), or trivial. \autoref{tab:acc_qpr} records the accuracy and number of parameters of MZeQAS against QCNN and Strongly Entangling Layers (SEL) ansatz  on this task. The number of parameters for QCNN varies with the number of qubits: QCNN with 1 layer uses 195 parameters for the 8-qubit task and 225 parameters for the 10-qubit task, while QCNN with 9 layers uses 1635 and 1920 parameters, respectively.. In the 8-qubit clustering task, MZeQAS matches the performance of QCNN with 9 layers, achieving 88.2\% accuracy while using substantially fewer parameters. Meanwhile, on the 10-qubit task, MZeQAS substantially outperforms QCNN with 1 layer and achieves accuracy close to QCNN with 9 layers, while using two-third the number of parameters. This result shows that our method can achieve similar performance to that of hand-crafted architecture.

\subsection{Hybrid-classic architecture}

\begin{table}[ht]
\centering
\scriptsize
\caption{Comparison of classical CNN, manually designed QT models, and our searched QT architecture on MNIST and CIFAR-10. \textsuperscript{*}}
\label{tab:hybrid_classic_mnist_cifar10}
\begin{tabular}{l|c|c|c|c}
\hline
\textbf{Model} & \textbf{Topology} & \textbf{Mapping model} & \textbf{Accuracy} & \textbf{Trainable params} \\
\hline
\multicolumn{5}{c}{\textbf{MNIST}} \\
\hline
Classical CNN 
& \multicolumn{2}{c|}{2 conv. layers, max pooling, 192-20, 20-10} 
& 95.5 
& 6690 \\

\hline
QT-1 
& 1 U3-CU3 gate layer 
& 14-4, 4-20, 20-4, 4-1 
& 33.5 
& 327 \\
QT-4 
& 4 U3-CU3 gate layers 
& 14-4, 4-20, 20-4, 4-1 
& 56.0
& 561 \\
QT-7 
& 7 U3-CU3 gate layers 
& 14-4, 4-20, 20-4, 4-1 
& 82.0 
& 795 \\
QT-10 
& 10 U3-CU3 gate layers 
& 14-4, 4-20, 20-4, 4-1 
& 85.3 
& 1029 \\
QT-13 
& 13 U3-CU3 gate layers 
& 14-4, 4-20, 20-4, 4-1 
& 90.7 
& 1263 \\
QT-16 
& 16 U3-CU3 gate layers 
& 14-4, 4-20, 20-4, 4-1 
& 91.8 
& 1497 \\
\hline
MZeQAS 
& searched QNN 
& 14-4, 4-20, 20-4, 4-1 
& \textbf{94.4} 
& \textbf{786} \\
\hline

\multicolumn{5}{c}{\textbf{CIFAR-10}} \\
\hline
Classical CNN 
& \multicolumn{2}{c|}{2 conv. layers, max pooling, 2048-128, 128-64, 64-10} 
& 62.58 
& 285226 \\
\hline
QT-19 
& 19 U3-CU3 gate layers 
& 20-40, 40-200, 200-40, 40-1 
& 31.58 
& 19287 \\
QT-38 
& 38 U3-CU3 gate layers 
& 20-40, 40-200, 200-40, 40-1 
& 33.12 
& 21543 \\
QT-57 
& 57 U3-CU3 gate layers 
& 20-40, 40-200, 200-40, 40-1 
& 39.24 
& 23619 \\
QT-76 
& 76 U3-CU3 gate layers 
& 20-40, 40-200, 200-40, 40-1 
& 51.06 
& 25785 \\
QT-95 
& 95 U3-CU3 gate layers 
& 20-40, 40-200, 200-40, 40-1 
& 63.10 
& 27951 \\
\hline
MZeQAS 
& searched QNN 
& 20-40, 40-200, 200-40, 40-1 
& \textbf{63.21} 
& \textbf{25950} \\
\hline
\end{tabular}
\end{table}


\autoref{tab:hybrid_classic_mnist_cifar10} compares MZeQAS with manually designed QT models on MNIST and CIFAR-10. On MNIST, the manually designed QT models improve as the number of U3-CU3 gate layers increases, reaching 91.8\% accuracy with 1497 trainable parameters. In contrast, MZeQAS achieves a higher accuracy of 94.4\% while using only 786 trainable parameters, showing that the searched architecture is both more accurate and more parameter-efficient than all manually designed MNIST QT baselines. On CIFAR-10, the best manually designed QT model achieves 63.10\% accuracy using 27951 trainable parameters. MZeQAS further improves the accuracy to 63.21\% while reducing the number of trainable parameters to 25950. Furthermore, comparing to QT-76, MZeQAS shows more than 10\% accuracy improvement while using approximately similar number of parameter. These results demonstrate that MZeQAS can identify more effective and compact circuit than manually designed QT models across both datasets.

\begingroup
\renewcommand{\thefootnote}{*}
\footnotetext{QT and CNN results are taken from the public repository: \url{https://github.com/Hon-Hai-Quantum-Computing/QuantumTrain/tree/chenyu_dev}.}
\endgroup

\section{Conclusion}
This paper introduces MZeQAS, a novel framework for quantum neural architecture search that efficiently identifies high-performing quantum circuits for classification tasks. The framework integrates a zero-shot proxy of the loss function into a MCTS search strategy, enabling effective evaluation of candidate circuits without costly iterative training. To achieve this, we propose AMES, a zero-shot evaluation function for parameterised quantum circuits. AMES is inspired by the minimum eigenvalue of the Gram matrix in overparameterised QNNs and is augmented with a coefficient proportional to the number of circuit parameters, thereby encouraging the exploration of highly expressive circuit architectures. Using AMES as the evaluation metric, we develop an MCTS-based search algorithm enhanced by batch sampling to further improve scalability. Experimental results on the QPR and image classification datasets demonstrate that MZeQAS achieves superior accuracy compared to traditional quantum architecture search methods while requiring significantly less search time. The hybrid-classic architecture shows that our framework can be extended beyond These results highlight the effectiveness of incorporating zero-shot evaluation and structured search in quantum neural architecture design. For future work, we plan to extend our framework to handle larger-scale QNNs and explore its adaptation for Variational Quantum Eigensolver tasks.

\bibliographystyle{unsrtnat}
\bibliography{example_paper}

\cleardoublepage

\newpage
\appendix
\onecolumn

\section{Additional numerical results}

\begin{table}[h]
\centering
\small
\caption{Comparison between GA with weight sharing and AMES on the MNIST datasets.}
\label{tab:surrogate_function}
\renewcommand{\arraystretch}{1.3} 
\begin{tabular}{p{2.5cm}|c|c|c}
\hline
\textbf{Method} & \textbf{MNIST-2} & \textbf{MNIST-4} & \textbf{MNIST-10} \\ \hline
{GA (QuantumNAS)} & 84.0 & 70.5 & 51.0 \\ \hline
{GA (QuantumNAS) w. AMES} & 84.7 & 76.3 & 51.4 \\ \hline
{MCTS w. AMES (MZeQAS)} & \cellcolor{gray!30}\textbf{85.6} & \cellcolor{gray!30}\textbf{77.3} & \cellcolor{gray!30}\textbf{53.4} \\ \hline
\end{tabular}
\end{table}
\paragraph{MCTS versus Genetic Algorithm} 
Table~\ref{tab:surrogate_function} highlights the performance advantage of our MCTS-based approach over traditional Genetic Algorithm (GA). MCTS leverages a tree-based search strategy with an explicit exploration–exploitation trade-off, guided by AMES, enabling more focused discovery of high-performing architectures with fewer evaluations.  Furthermore, MCTS incrementally builds knowledge about the search space, reusing information from previous circuit candidates to guide subsequent decisions. As shown in Table~\ref{tab:surrogate_function}, MCTS achieves consistently higher accuracy across all MNIST datasets (MNIST-2, MNIST-4, and MNIST-10) compared to GA, demonstrating its superior ability to identify high-quality circuit architectures. Furthermore, using AMES instead of weight-sharing provides a more reliable performance estimation during the search, thereby reducing the dependency on full retraining and improving GA's overall robustness.

\begin{table}[h]
\centering
\caption{Experimental result of MZeQAS when run on ibmq\_quito}
\label{tab:ibm_res}
\begin{tabular}{lcc}
\hline
\textbf{Dataset} & \textbf{MZeQAS} & \textbf{MZeQAS (IBM Topology)} \\
\hline
MNIST-2  & 85.6 & 83.4 \\
MNIST-4  & 77.3 & 78.3 \\
\hline
\end{tabular}
\end{table}

\paragraph{Evaluation on IBM topology} In this experiment, we evaluate the performance of MZeQAS using the topology of IBM’s 5-qubit quantum device (ibmq\_quito). The connectivity of ibmq\_quito is visualised in Figure~\ref{fig:ibm}. During the expansion phase, any selected two-qubit gates must satisfy the device’s connectivity constraints before being added. As shown in \autoref{tab:ibm_res}, MZeQAS maintains performance comparable to the original configuration even under a different topology. The accuracy difference is minimal, with only 2.2 percentage points for MNIST-2, while MNIST-4 even improves slightly under the IBM topology.

\begin{figure}[H]
    \centering

    \begin{subfigure}{\textwidth}
        \centering
        \includegraphics[width=\linewidth]{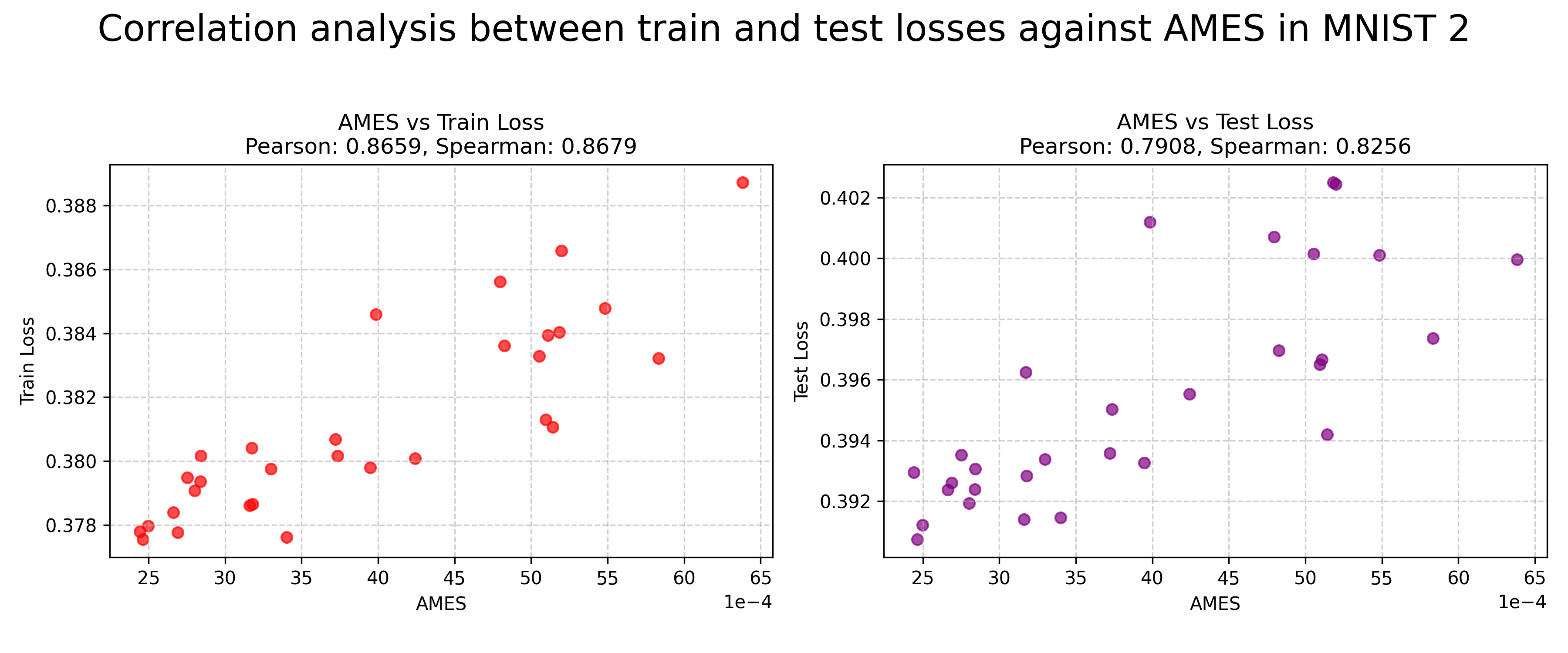}
        \caption{Correlation of AMES with circuits having 40 - 100 randomly generated gates}
    \end{subfigure}
    \hfill
    \begin{subfigure}{\textwidth}
        \centering
        \includegraphics[width=\linewidth]{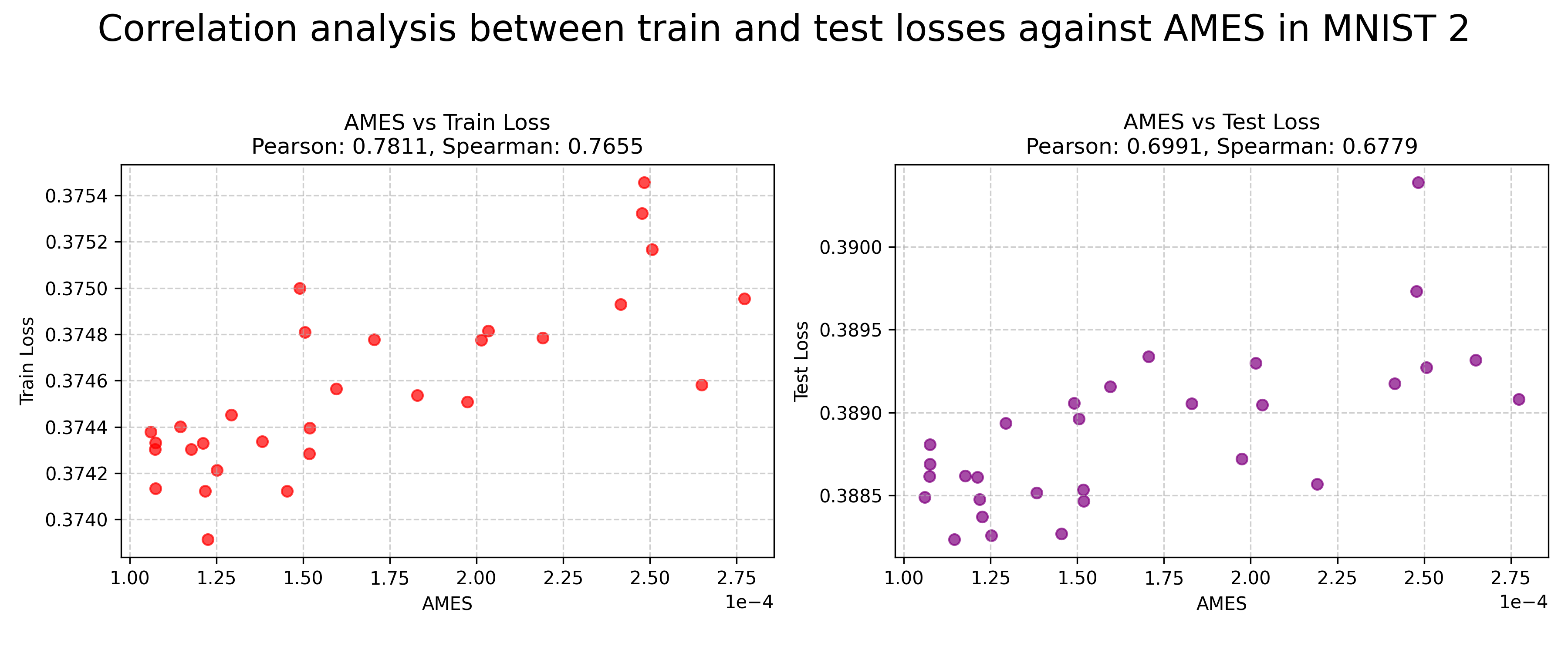}
        \caption{Correlation of AMES with circuits having 300 - 1000 randomly generated gates}
    \end{subfigure}

    \caption{Correlation analysis between train and test losses against AMES in MNIST 2 classes}
    \label{fig:ames_mnist2}
\end{figure}

\paragraph{Analysis of the importance of the parameter coefficient} Table~\ref{tab:pc_impact} compares the accuracy of MZeQAS with and without incorporating $p_c$ in the AMES surrogate function across three datasets 
(MNIST-2, MNIST-4, and MNIST-10). The results show that including $p_c$ consistently improves accuracy on all datasets. Specifically, accuracy increases from 80.0\% to 85.6\% on MNIST-2, from 58.3\% to 77.3\% on MNIST-4, and from 41.3\% to 53.4\% on MNIST-10. These improvements demonstrate that accounting for the number of parameters, $p_c$, in the surrogate function effectively rewards larger circuits and leads to better-performing architectures. This behavior can be attributed to how the surrogate cost is defined: because $1/p_c$ decreases as $p_c$ increases, circuits with more parameters are assigned a lower zero-shot proxy cost.  However, this observation also highlights an important trade-off: while encouraging larger models can boost accuracy, it may lead to increased circuit depth and resource consumption on quantum hardware. 

\begin{table*}[h]
\centering
\caption{Comparison of MZeQAS accuracy on MNIST datasets with and without $p_c$ in AMES.}
\begin{tabular}{l|c|c|c}
\hline
\textbf{Method} & \textbf{MNIST-2} & \textbf{MNIST-4} & \textbf{MNIST-10} \\ \hline
MZeQAS w/o $p_c$ & 80.0 & 58.3 & 41.3 \\ \hline
MZeQAS w $p_c$ & \cellcolor{gray!30}\textbf{85.6} & \cellcolor{gray!30}\textbf{77.3} & \cellcolor{gray!30}\textbf{53.4} \\ \hline
\end{tabular}
\label{tab:pc_impact}
\end{table*}

\paragraph{Correlation analysis AMES and downstream task loss}

Figure~\ref{fig:ames_mnist2} measures the AMES value and the corresponding train and test losses after training the circuit for 100 iterations. It shows that AMES is positively correlated with both training and testing losses on the MNIST-2 classification task. As the AMES value increases, both the train and test losses consistently increase across different circuit depths. For circuits with 40 - 100 randomly generated gates, the correlations are particularly strong, with Pearson coefficients of 0.8659 for train loss and 0.7908 for test loss. A similar positive relationship is observed for deeper circuits containing 300 - 1000 gates, where the Pearson coefficients remain significantly positive at 0.7811 and 0.6991 for train and test losses, respectively. These results indicate that larger AMES values are associated with poorer performance, indicating that AMES is an informative surrogate for QAS task.

\section{Technical proof}
\label{appendix:proof}
\begin{definition}[Periodic ansatz]
\label{appendix: periodic ansatz}
Let $H \in \mathbb{C}^{d \times d}$ be a fixed nonzero traceless Hermitian operator. 
A periodic ansatz with $p$ parameters is defined by
\begin{equation}
U(\vt{\theta}) := U_p \exp{(-i \theta_p \vt{H})} \cdots U_1 \exp{(-i \theta_1 \vt{H})} \, U_0,
\end{equation}
where $\vt{\theta} = (\theta_1,\ldots,\theta_p)$ and the unitaries  $\{U_i\}_{i=0}^p \subset SU(d)$ are independently drawn according to the Haar measure on $SU(d)$.
\end{definition}

Proof for \autoref{theorem: t=inf}: 
\begin{proof}
Let $U(\boldsymbol{\theta}) := U_p\exp(-i\theta_p \vt{H})...U_1\exp(-i\theta_1 \vt{H})U_0$ be the periodic ansatz as defined in \cite{you2023analyzing} (restated in Definition~\ref{appendix: periodic ansatz}) and the parameterized measurement is $M(\vt{\theta}) = U^\dagger(\vt{\theta})M_0U(\vt{\theta})$. Using Taylor expansion on $\exp(-i\theta_j \vt{H})$, we have
\begin{equation}
    \exp(-i\theta_j\vt{H}) \approx I - i\theta_j\vt{H} O(\theta^2). 
\end{equation}
Subsequently, 
\begin{equation}
    U(\vt{\theta}) = (U_pU_{p-1}...U_0) - i\sum_{j=1}^{p}(U_p...U_{j+1})\vt{H}(U_j...U_0)\theta_j + O(\theta^2).
\end{equation}
Let $V:=U_pU_{p-1}...U_0$ and $A_j:=- i(U_p...U_{j+1})\vt{H}(U_j...U_0)$. The above equation becomes $V + \sum_{j=1}^{p}A_j\theta_j + O(\theta^2)$. Similarly, $U^\dagger = V^\dagger + \sum_{j=1}^{p}A_j^\dagger\theta_j + O(\theta^2) $. Using these result, we can decompose $M(\vt{\theta})$ into sum of random variables as follows:
\begin{equation}
    M(\vt{\theta}) \approx V^\dagger M_0 V + \sum_{j=1}^{p} \theta_j (V^\dagger M_0 A_j + A^\dagger_j M_0 V) + O(\theta^2).
\end{equation}
Thus, for every density matrix $\rho$
\begin{equation}
    [M(\vt{\theta}), \rho] \approx [V^\dagger M_0 V, \rho] + \sum_{j=1}^{p} \theta_j [(V^\dagger M_0 A_j + A^\dagger_j M_0 V), \rho] + O(\theta^3).
\end{equation}
Since $\theta \sim \mathcal{N}(0, \vt{I}/p)$ then the term $O(\theta ^ 3)$ vanishes as $p$ tends to infinity. Let $G_{j,k}:=[(V^\dagger M_0 A_j + A^\dagger_j M_0 V), \rho_k]$ and $G_{0,k}:=[V^\dagger M_0 V, \rho_k ]$, then $[M(\vt{\theta}), \rho_k] = G_{0,k} + \sum_{j=1}^p \theta_j G_{j,k}$. Now, we will derive the value of $(K_{\mathrm{asym}}(t))_{kj}$ defined in section \ref{sec:zero-shot}:
\begin{equation}
    (K_{\mathrm{asym}}(t))_{kj} \approx -\tr{G_{0,k}G_{0, j}} - \sum_{l=1}^p \theta_l \tr{G_{l, k}G_{0, j} + G_{0, k}G_{l, j}} - \sum_{l,m=1}^{p} \theta_l \theta_m\tr{ G_{l,k} G _{m, j}}.
\end{equation}
Denote $L_{k,j}(\vt{\theta}):=-\sum_{l=1}^p \theta_l \tr{G_{l, k}G_{0, j} + G_{0, k}G_{l, j}}$ and $W_{k,j}(\vt{\theta}):=- \sum_{l,m=1}^{p} \theta_l \theta_m\tr{ G_{l,k} G _{m, j}}$. The expectation value of each entry in $L$ is 
\begin{equation}
\label{eq:limL}
\mathbf{E}(L_{k,j}(\vt{\theta})) = -\sum_{l=1}^p \tr{G_{l, k}G_{0, j} + G_{0, k}G_{l, j}} \mathbf{E}(\theta_l) = 0.
\end{equation}
Meanwhile, we can rewrite $W(\theta) = \vt{\theta}^\intercal B \vt{\theta}$. Using the Hanson-Wright inequality, 
\begin{equation}
    \Pr(W - \mathbf{E}(W) \geq t) \leq 2 \exp \left(-c \min \left(\frac{p^2 t^2}{\| B\|^2_F}, \frac{pt}{\| B\|_{op}} \right) \right)
\end{equation}
Hence, $\Pr(W - \mathbf{E}(W) \geq t)  \leq 2 \exp(-cp).$ Therefore
\begin{equation}
    \label{eq:limW}
    \lim_{p\to\infty}|W-\mathbb{E}(W)| \to 0.
\end{equation} As a result, using the result from Eq.~\ref{eq:limL} and Eq.~\ref{eq:limW}, each term in $(K_{\mathrm{asym}}(t))_{kj}$ is either remain constant or approaching its expected value as $p$ tends to infinity. Therefore, we have 

 \begin{equation}
    \lim_{p\to\infty}|(K_{\mathrm{asym}}(t))_{kj} - 
    (\mathbf{K}_{\mathrm{asym}}^{\infty})_{kl})| \to 0.
\end{equation}
\end{proof}

\begin{corollary}
\label{cor:lambda-equal}
Under the entry-wise convergence of Gram matrix in \autoref{theorem: t=inf} then 
$$ 
\lim_{p \to \infty} \| \lambda_{0}(t) - \lambda_{\min} \| = 0.  
$$  where $\lambda_0(t)$ and $\lambda_{min}$ are the minimum eigenvalue of $\mathbf{K}_{\mathrm{asym}}(t)$ and $\mathbf{K}_{\mathrm{asym}}^{\infty}$ respectively.
\end{corollary}

\begin{proof}
In \autoref{theorem: t=inf}, we showed that the Gram matrix converges entrywise, i.e.  
$$\lim_{p \to \infty} \big| ( \mathbf{K}_{\mathrm{asym}}(t) )_{kl} - ( \mathbf{K}_{\mathrm{asym}}^{\infty} )_{kl} \big| = 0 \quad \forall k,l.$$
This entrywise convergence implies that the difference matrix $\mathbf{K}_{\mathrm{asym}}(t) - \mathbf{K}_{\mathrm{asym}}^{\infty}$ converges to zero in Frobenius norm. Therefore, its spectral norm satisfies  
$$
\| \mathbf{K}_{\mathrm{asym}}(t) - \mathbf{K}_{\mathrm{asym}}^{\infty} \|_2 \to 0.  
$$ 
Since both matrices are symmetric, we use Weyl’s inequality, which yields  
$$
\big| \lambda_{0}(t) - \lambda_{\min} \big| \le \| \mathbf{K}_{\mathrm{asym}}(t) - \mathbf{K}_{\mathrm{asym}}^{\infty} \|_2.  
$$
Consequently,  
$$ 
\lim_{p \to \infty} \| \lambda_{0}(t) - \lambda_{\min} \| = 0.  
$$  
\end{proof}

\section{Experiment setup}
\label{sec: expsetup}

\begin{table*}[h]
\centering
\scriptsize
\caption{Dataset and encoding description}
\label{tab:dataset}
\begin{tabular}{cccclc}
\hline
\textbf{Task} & \textbf{Classes} & \textbf{Input Size} & \textbf{\#qubit} & \textbf{Encoder (Mapping)} \\ \hline
MNIST-10 & 0--9 & 6$\times$6 & 10 & 10RY, 10RZ, 10RX, 6RY \\ \hline
MNIST-4  & 0,1,2,3 & 4$\times$4 & 4 & 4RY, 4RZ, 4RX, 4RY \\ \hline
MNIST-2  & 3,6 & 4$\times$4 & 4 & 4RY, 4RZ, 4RX, 4RY \\ \hline
Fashion-4 & t-shirt, trouser, pullover, dress & 4$\times$4 & 4 & 4RY, 4RZ, 4RX, 4RY \\ \hline
Fashion-2 & dress, shirt & 4$\times$4 & 4 & 4RY, 4RZ, 4RX, 4RY \\ \hline
Vowel-4   & hid, hId, hEd, hAd & 10 & 4 & 4RY, 4RZ, 2RX \\ \hline
\end{tabular}
\end{table*}

\begin{table*}[ht]
\centering
\scriptsize
\caption{Hyperparameters for MCTS and circuit training}
\begin{tabular}{l|l|c|c}
\toprule
\textbf{Category} & \textbf{Hyperparameter} & \textbf{Symbol} & \textbf{Value} \\
\midrule
\multirow{5}{*}{\textbf{Searching Config}} 
    & MCTS's maximum evaluation     & $FE$       & 1000  \\
    & MCTS's widening coefficient   & $\alpha$   & 0.5   \\
    & MCTS's initial decay rate     & $\beta_0$& 0.1   \\
    & Maximum tree depth                  & $D$        & 16    \\
    & Number of children for each node    & $B$        & 4 \\
    & U3, CU3 choosing probabilities      & $P$      & [0.5, 0.5] \\
    & XX, YY, ZZ choosing probabilities   & $P$      & [0.333, 0.333, 0.333] \\
\midrule
\multirow{2}{*}{\textbf{Circuit training Settings}}
    & Learning rate                       & $\eta$     & 0.005 \\
    & Training epochs                     & $T$        & 200   \\
\bottomrule
\end{tabular}
\label{tab:hyperparameters}
\end{table*}

\begin{table}
    \centering
    \caption{Hyperparameter settings for QuantumNAS, PWMCTS and EQNAS}
    \label{tab:placeholder}
    \begin{tabular}{l|l|c}
        Algorithm & Parameter & Value \\ \hline
        \multirow{6}{*}{QuantumNAS} 
        & Number of generation & 40 \\
        & Population size & 40 \\
        & Parent population size&  10 \\
        & Mutation population size & 20 \\
        & Mutation rate & 0.4 \\
        & Crossover population size & 10 \\ \hline
        \multirow{13}{*}{PWMCTS} 
            & Roll-out steps                 & 0   \\
            & Action-by-action              & 5\% \\
            & UCB coefficient & 0.4 \\
            & Progressive widening coefficient & 1   \\
            & Progressive widening exponent   & 0.3 \\
            & Adding probability             & 0.5 \\
            & Swapping probability           & 0.2 \\
            & Changing probability           & 0.2 \\
            & Deleting probability           & 0.1 \\
            & Maximum circuit depth  & 20  \\ \hline
        \multirow{2}{*}{EQNAS} 
        & Number of generation & 20\\
        & Mutation probability& 0.2 \\ \hline
    \end{tabular}

\end{table}

\paragraph{Image classification}
We evaluate the performance of MZeQAS on the image classification dataset MNIST \cite{deng2012mnist}, Fashion-MNIST \cite{xiao2017fashion} and Vowel \cite{vowel} dataset. A breakdown of dataset and their corresponding label, input size and encoder is provided in Table \ref{tab:dataset}. Specifically, for MNIST and Fashion-MNIST, we sample 5000 images as training set, 3000 as validation set and 300 as the test set. In the MNIST and Fashion-MNIST, we use average pooling to reduce the original image to the desired number of dimension, equivalent to the number of qubits as denoted in Table \ref{tab:dataset}. We then use these vector as rotations angle for encoding gates in Table \ref{tab:dataset}. For the Vowel \cite{vowel} dataset, Principal Component Analysis (PCA) was utilized to identify and extract the 10 most salient features. Pauli-Z was selected as the measurement basis. In binary classification scenarios, the measurement outcomes from qubits 1 and 2 were aggregated, as were those from qubits 3 and 4, before being fed into a Softmax function to produce the corresponding class probabilities. MZeQAS is compared against QuantumNAS \cite{wang2022quantumnas}, PWMCTS \cite{lipardi2025quantum} and EQNAS \cite{li2023eqnas}. We implemented MZeQAS using the TorchQuantum library \cite{wang2022quantumnas}. The comparison with EQNAS uses two-qubit gates: \textit{XX}, \textit{YY}, and \textit{ZZ}. In EQNAS, We maintain the same data-splitting strategy as previous experiment while using the encoding approach from \cite{li2023eqnas}. Every founded circuit is trained using Adam optimizer and cosine learning rate scheduler. In order to reduce the computing time, during the search, we randomly sample 100 images and use them to construct the Gram matrix. 
\begin{figure}[h]
    \centering
    \includegraphics[width=0.4\columnwidth]{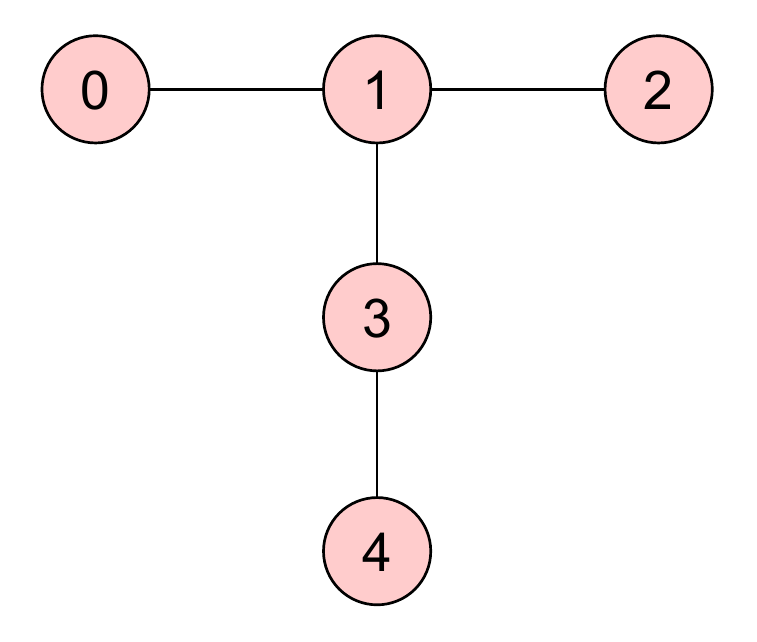}
    \caption{The topology IBM's 5 qubits device ibmq\_quito}
    \label{fig:ibm}
\end{figure}

\paragraph{Quantum Phase Recognition} In this task, the search space is $U3$, $RY$, $RZ$ and $CNOT$. Moreover, we also initialize the MCTS root with 1 layer of QCNN. For each QPR task, we generate 1000 and 100 samples as training and test datasets respectively. We use Adam optimizer with learning rate 0.001 and train for 2000 steps. We perform softmax on the measurement of qubit 1, 3, 5, 7 as prediction for 4 class classification. In order to reduce the computing time, during the search, we randomly sample 10 instances and use them to construct the Gram matrix. Furthermore, 

\paragraph{Hybrid classic architecture search task}
The Quantum-Train (QT) framework \cite{liu2025quantum}, uses a QNN as a compact generator of the weights of a classical neural network, rather than using the QNN directly for inference. Given a classical model with $\tau$ parameters, QT constructs an $d=\lceil \log_2 \tau\rceil$-qubit parameterized quantum circuit whose measurement probabilities are mapped, through a small classical mapping network, into the full parameter vector of the target classical model. During training, both the QNN parameters and mapping-network parameters are optimized using the task loss, while the classical network processes the actual input data and produces predictions. This design reduces the number of trainable parameters from $\tau$ to $O(\mathrm{polylog}(\tau))$, avoids difficult large-scale quantum data encoding, and allows the final trained model to run entirely on classical hardware at inference time. In our work, we use MZeQAS to search for effective QNN architectures within the QT framework. In our experiments, the MZeQAS search space consists of $U3$ and $CU3$ quantum gates. For MNIST, we evaluate the models using 60000 training images and 10000 test images, and MZeQAS is searched from scratch. For CIFAR-10, we evaluate the models using 50000 training images and 10000 test images. Since the CIFAR-10 search space is larger and the optimization is more challenging, we initialize MZeQAS from the QT-76 architecture as a warm start. During the warm-start search, MZeQAS is allowed to prune and modify the original circuit, so the final architecture is not constrained to preserve the full QT-76 structure or its 25785 trainable parameters. As a result, the searched CIFAR-10 architecture can achieve a different parameter count while retaining only the components that are beneficial for performance.

The computational experiments were conducted on three devices: a server with a 96-core Intel(R) Xeon(R) CPU @ 2.00GHz and 335GB RAM, a machine with an Intel(R) i7-10750H(R) CPU @ 2.60GHz and 8GB RAM, and a GPU device with an NVIDIA RTX 5070 GPU with 16GB VRAM.


\section{Pseudocode of MZeQAS}
Algorithm~\ref{alg: main} depicts the pseudocode of MZeQAS's main workflow. This is a more detailed description of the framework in figure~\ref{fig:mcts}. AMES is primarily use in Simulation step.
\begin{algorithm}
\caption{MZeQAS}
\small
\label{alg: main}
\begin{algorithmic}[1]
\State \textbf{Input:} Max tree depth: $D$, Gate selecting probability: $P$, Gate set: $S$, Maximum evaluation: $FE$, Fitness function: $g(.)$, Initial UCB's initial decay rate: $\beta_0$, Number of children for each node: $B$
\State \textbf{Initialize:} Root node $r$ with empty circuit
\State Set $q_{\min} \gets \infty$, $q_{\max} \gets -\infty$
\State $fe\gets0$
\While{$fe$ $< FE$}
    \State // \textbf{Selection}
    \State $c \gets r$
    \While{$c$ is not a leaf and its depth is less than $D$}
        \State Compute $\beta$ using equation~\ref{eq: decay}
        \State Select child $\hat{c} \in \text{children}(c)$ that maximizes equation~\ref{eq:uct}
        \State $c \gets \hat{c}$
        \State // \textbf{Progressive Widening}
        \If{$\lfloor N(c)^\alpha \rfloor > |\text{children}(c)|$}
        \If{Random() $< 0.5$}
            \State Randomly delete gate(s) from $c$ to create $c'$
        \EndIf
        \Else
            \State Randomly add a gate from S to $c$ following the probability distribution $P$ to create $c'$
        \EndIf
        \State Run simulation and backpropagation for $c'$
    \EndWhile

    \State // \textbf{Expansion}
    \For{$i = 1$ to $B$}
        \State Sample $H \sim \mathcal{U}(1, L)$
        \State Sample $H$ gates from $S$ following the probability distribution $P$, denote as $c_{new}$
        \State $\hat{c} \gets$ $c \cup c_{new}$
        \State Add $\hat{c}$ to children($c$)
    \EndFor

    \State // \textbf{Simulation}
    \For{$\hat{c} \in children(c)$}
        \State Compute fitness $g(\hat{c})$ using AMES
        \State $Q(\hat{c}) \gets -g(\hat{c})$
        \State $N(\hat{c}) \gets 1$
        \State $fe \gets fe + 1$
        \State $q_{max} \gets max(q_{max}, Q(\hat{c})), q_{min}\gets min(q_{min}, Q(\hat{c}))$
    \EndFor
    
    \State // \textbf{Backpropagation}
    \While{$c$ is not the root node $r$}
        \State $Q(c) \gets \max\limits_{\hat{c} \in \text{children}(c)} Q(\hat{c})$
        \State $N(c) \gets \sum\limits_{\hat{c} \in \text{children}(c)} N(\hat{c})$
        \State $c \gets Father(c)$
    \EndWhile
\EndWhile
\State \textbf{Output:} Circuit $c^*$ with minimum $g(c^*)$
\end{algorithmic}
\end{algorithm}

\section{Limitation discussion and Outlook}
\label{sec:disscusion}
In this work, we presented MZeQAS, a zero-shot quantum architecture search framework that leverages the convergence behavior of the Gram matrix in overparameterized QNNs to construct a training-free surrogate objective. By combining the proposed AMES surrogate with MCTS, MZeQAS is able to efficiently explore large quantum circuit search spaces while avoiding the computational overhead associated with repeated circuit training. For fair comparison, the image-classification experiments in this paper adopt a fixed data encoding strategy across all search algorithms and focus exclusively on optimizing the variational circuit architecture. However, the choice of data encoding also influences the expressivity, trainability, and inductive bias of QNNs. In particular, different encoding schemes may interact with circuit architectures in non-trivial ways, potentially affecting the optimization landscape and generalization performance. Consequently, an important direction for future work is to extend MZeQAS toward a joint optimization framework that simultaneously searches for both the quantum circuit architecture and the feature-encoding strategy. Such an approach could enable the discovery of task-specific end-to-end quantum models that are better aligned with the underlying data distribution. Another limitation is that the proposed surrogate is primarily designed for supervised QNN pipelines trained using gradient-based optimization. Future work may therefore investigate extending the surrogate framework to other important quantum computing paradigms, including the Variational Quantum Eigensolver (VQE) and the Quantum Approximate Optimization Algorithm (QAOA). In these settings, the optimization objectives differ substantially from supervised learning losses. Developing zero-shot proxies for these algorithms could significantly reduce the cost of architecture discovery in quantum chemistry and combinatorial optimization applications. 

\end{document}